\documentclass[10pt,twocolumn]{article} 
   
\usepackage{amsmath,amsfonts,amssymb}
\usepackage{graphics,graphicx,color,psfrag}
\usepackage{enumerate}
\usepackage{ifthen}
\usepackage{algorithm}
\usepackage{algorithmic}
\usepackage[margin=1in]{geometry}  \usepackage[pdfauthor={Hanumant Singh Shekhawat},
            pdftitle={Frequency truncated discrete-time system norm}]{hyperref}

\usepackage{mypack} 

\allowdisplaybreaks
\makeatletter
\DeclareRobustCommand{\qed}{  \ifmmode   \else \leavevmode\unskip\penalty9999 \hbox{}\nobreak\hfill
  \fi
  \quad\hbox{\qedsymbol}}
\newcommand{\openbox}{\leavevmode
  \hbox to.77778em{  \hfil\vrule
  \vbox to.675em{\hrule width.6em\vfil\hrule}  \vrule\hfil}}
\newcommand{\qedsymbol}{\openbox}
\newenvironment{proof}[1][\proofname]{\par
  \normalfont
  \topsep6\p@\@plus6\p@ \trivlist
  \item[\hskip\labelsep\itshape
    #1.]\ignorespaces
}{  \qed\endtrivlist
}
\newcommand{\proofname}{Proof}
\newcommand{\Gt}{\ensuremath{\tv{G}}} \newcommand{\Gf}{\ensuremath{G}} \makeatother

\begin{document}

\title{Frequency truncated discrete-time system norm\footnote{The material in this paper was partially presented at 22th International
	Symposium on Mathematical Theory of Networks and Systems (MTNS 2016), 
	July 11-15, 2016, Minneapolis, USA. }} 

\author{Hanumant Singh Shekhawat \\
Indian Institute of Technology Guwahati, Guwahati, India\\
h.s.shekhawat@iitg.ac.in  \\
}                                        
\date{}             

\maketitle

\begin{abstract}                          Multirate digital signal processing and model reduction applications require computation of the frequency truncated norm of a discrete-time system. This paper explains how to compute the frequency truncated norm of a discrete-time system. To this end, a much-generalized problem of integrating a transfer function of a discrete-time system given in the descriptor form over
an interval of limited frequencies is also discussed along with its computation.
\end{abstract}

\section{Introduction}
The  frequency truncated discrete-time system norm of a linear discrete-time-invariant $\Gt$ (with the 
transfer function $\Gf(z)$ in $z$-domain)
 defined as
\begin{align}\label{dftn}
\|\Gt\|^2_{[\theta_1, \theta_2]} 
   := \frac{1}{2\pi}\tr\int_{\theta_1}^{\theta_2} \Gf^\sim(\ejth) \Gf(\ejth) \dth, 
\end{align}
where the conjugate system $\Gf^\sim(z) := \Gf^*(\bar{z}^{-1})$ ($*$ is the adjoint operation and $\bar{z}$ is complex conjugate of $z$).
The need for the  frequency truncated discrete-time system norm arises naturally in the multi-rate discrete signal processing. For example, consider
a simple setup  of multi-rate discrete signal processing as shown in \figref{dtn:fig:setup}.
\begin{figure}[htbp]
	\centering
		\includegraphics[width=75mm]{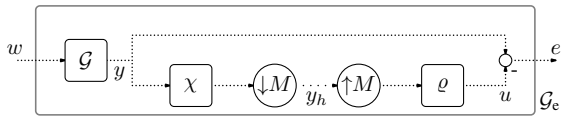}
	\caption{A setup  for multi-rate discrete signal processing}
	\label{dtn:fig:setup}
\end{figure}
Here, the input discrete signal ${y}$ is real and represented as the output of a system $\Gt$ driven by the discrete white Gaussian noise with zero mean. 
 The transfer function (in $z$-domain) of $\Gt$ is represented by $\Gf(z)$.
Hence, the power spectral density of  ${y}$ is given by $|\Gf(\ejth)|^2$ for
 all frequencies $\theta \in [-\pi,\pi]$.  $\chi$ is the analysis filter whose output is down-sampled by a factor $M$. The down-sampled  output
is again upsampled by a factor $M$ followed by a synthesis filter $\varrho$. The reconstructed output ${u}$ is compared with the input signal ${y}$.
The aim is to design both the analysis and synthesis filter   given $\Gf(z)$ in such a way that 
time averaged mean square error 
\begin{align*}
J = \lim_{N\to \infty}  \frac{1}{2N+1} \sum_{n=-N}^{N}E(e^2[n])
\end{align*}
 is minimised \cite{tsat}. Here, $E$ is the expectation operator. Assume that the $\Gt$ is stable and $\Gf(\ejth)$ is dominant in the frequency band $[-\frac{\pi}{M}, \frac{\pi}{M}]$ that means  $|\Gf(\e^{\jth_1})| > |\Gf(\e^{\jth_2})|$ if $\theta_1 \in [-\frac{\pi}{M}, \frac{\pi}{M}]$ and $\theta_2 \in [-\pi,\pi]\backslash [-\frac{\pi}{M}, \frac{\pi}{M}]$. In this case, optimal
 synthesis and analysis filter give   
 the error (see \cite{tsat} for details)
\begin{align}\label{dtn:me}
J =  \|\Gt\|_{\Ltwosys}^2 
                         - \frac{1}{2\pi}\tr\int_{-\frac{\pi}{M}}^{\frac{\pi}{M}}\Gf^\sim(\ejth) \Gf(\ejth) \dth.
\end{align}
Here, $\|\Gt\|_{\Ltwosys}:=\|\Gt\|_{[-\pi, \pi]}$ represents the $\Ltwosys$ norm of the system.  
Thus, we can see that truncated system norms naturally arises in multi-rate discrete signal processing.
For a general $\Gf(\ejth)$, computation of \eqref{dftn} is needed \cite{tsat}. 

It is further assumed that the discrete-time system $\Gt$ in \eqref{dftn} is 
linear discrete-time-invariant (LDTI) and its 
transfer function $\Gf(z)$ is a proper rational transfer function with real coefficients. These type of
systems can be represented  in state space  as $\Gf(z) = D+ C(z I-A)^{-1}B$ 
where $A,B,C$ and $D$ are 
real matrices. 
For simplicity of exposition, it is assumed that 
$D=0$ throughout this paper.
It will be shown later in the paper that evaluation of the  frequency truncated discrete-time system norm 
is a special case of the generalized problem of integration of a transfer function given in the descriptor form i.e\
\begin{align}\label{dtn:des}
\int_{\theta_1}^{\theta_2} C(\ejth E-A)^{-1}B \dth
\end{align}
where $A,B,C$, and $E$ are real matrices.
The above integral is also helpful in the evaluation of the frequency-domain controllability and
the observability Grammian of a system with transfer function $C(z I-A)^{-1}B$~\cite{wang,horta,peter,ghafoor}.

An expression for the   frequency truncated discrete-time system norm  for a stable discrete-time system is given in \cite[Theorem 3.8]{peter}
which depends upon invertibility of the $A$ matrix. This paper provides a  modification 
which resolves this problem for a stable discrete-time system. 
This paper further generalizes the result of \cite[Theorem 3.8]{peter} and provides an expression for \eqref{dftn} with minimal restriction on the system poles. 
 Integral of a transfer function (of a discrete-time system) in the descriptor form and a numerically viable expression for its computation is also given in the paper.  As a by-product, similar results for a
continuous time system given in the descriptor form are
briefly mentioned.

\secref{dtn:sec:stable} contains an expression for the    frequency  truncated  discrete-time system norm for a stable system and 
\secref{dtn:sec:gen} contains an expression for the    frequency  truncated discrete-time system norm for a generic case.
Integration of a transfer function given in the descriptor form is discussed in \secref{dtn:sec:gen} and a method for its computation 
is given in \secref{dtn:sec:comput}. \secref{dtn:sec:cont} contains results related to frequency  truncated  norm
of a continuous time system given in the descriptor form.

\emph{Notation:} $\mb{R}$ and $\mb{C}$ denote the set of real and complex numbers respectively. $\cnR$ denotes the closed negative real axis (i.e.\ the negative real axis including zero).
$\parg(z)$ denotes the principal argument i.e.\ argument of the complex number $z$ in $(-\pi,\pi]$ and  $\wrap(\theta):=\parg(\ejth)$ for all $\theta \in \mb{R}$.
For square complex matrices $A$ and $E$, the matrix pencil $(A,E)$ is called regular if $\alpha E + \beta A$ is invertible for at-least one set of complex numbers $\alpha$
and $\beta$. An eigenvalue $\lambda \in \mb{C}$ of a
matrix pencil $(A,E)$ satisfies $\det(A-\lambda E)=0$. We define $\rj:=\sqrt{-1}$.
$\sigma (A,E)$ denotes spectrum  (the set of eigenvalues) of the matrix pencil $(A,E)$. A function $f$ is \emph{defined} on $\sigma(A)$ 
if it follows definition $1.1$ of \cite{higham}.  
If $A\in \mb{C}^{n \times n}$ does not have any eigenvalues on
the $\cnR$ then there is a unique logarithm $Q=\log(A)$ whose all eigenvalues lie in the  
open horizontal strip $\{ z \in \mb{C} | −\pi < \imag( z ) < \pi\}$ of the
complex plane \cite[Theorem 1.31]{higham}. The $Q$ is known as the \emph{principal logarithm}.
The Fr\'{e}chet derivative of a matrix function $f: \mb{C}^{n\times n} \to \mb{C}^{n\times n}$ 
at $A$ in the direction $X$ is denoted by $L_{f}(Q,X)$ \cite[\S 3.1]{higham}.

\section{Stable Case}\label{dtn:sec:stable}
A discrete-time system $\Gt$ is defined stable if $A$ is Schur (i.e.,\ all eigenvalues of $A$ are strictly inside the unit circle in the complex plane).  
It is well known that for a stable discrete-time system $\Gt$,
the squared $L^2$ norm is given by 
\begin{equation}\label{hhh}
  \|\Gt\|^2_{\Ltwosys} :=
   \frac{1}{2\pi}\tr\int_{-\pi}^{\pi}\Gf^\sim(\ejth) \Gf(\ejth) \dth = \tr B^TPB
\end{equation}
where $P$ is the unique solution of the discrete Lyapunov equation
\begin{equation}\label{dlyap}
  A^TPA -P + C^T C=0.
\end{equation}
The  frequency truncation discrete-time system norm can be calculated by  evaluating  
an anti-derivative (or primitive) $\int  \Gf^\sim(\ejth) \Gf(\ejth) \dth$ first. 
\begin{theorem}\label{dtn:thm:stable}
Let discrete-time system $\Gt$ be stable and
  strictly proper and $\Gf(\ejth)=C(\ejth I-A)^{-1}B$ with A, B, C real matrices.
  Then, an anti-derivative $\int  \Gf^\sim(\ejth) \Gf(\ejth) \dth$ equals
  \begin{align*}
    B^T P B \theta
                            + 2\imag (B^T P \log(I - \e^{-\jth} A)B) \notag
  \end{align*}
  where $P$ is the unique solution of~\eqref{dlyap} and $\log$ denotes
  the principal logarithm and $\theta\in [-\pi,\pi]$.
\end{theorem}
\begin{proof}
Consider function 
\begin{align*}
   f(z,\theta) = \begin{cases}
             -\rj z^{-1}\log(1 - \e^{-\jth} z) & \text{if } z\ne 0 \\
              \rj \e^{-\jth} & \text{if } z = 0
          \end{cases}
\end{align*}
For a given $\theta \in [-\pi,\pi]$, $f(z,\theta)$ is analytic in the open unit disk around zero (in the complex plane) as 
$1 - \e^{-\jth} z$ never lies on the closed negative real axis $\cnR$.
Clearly,
  \begin{align*}
  \pd{f(z,\theta)}{\theta} =   (\ejth -  z)^{-1}
  \end{align*}
 is also analytic (for a given $\theta \in [-\pi,\pi]$) in the open unit disk around zero.
Hence, it follows from \cite[Theorem 6.2.27]{horn}
  that $f(A,\theta)$ is an anti-derivative of  $(\ejth I-  A)^{-1}$.
 Also, note that
 \begin{align*}
 \Gf^\sim(\ejth )& \Gf(\ejth ) =  B^T (I-\ejth A^T )^{-1}PB \\
 &+ B^T PA (\ejth I-A)^{-1} B \\
 =&  B^TPB+ 	\ejth  B^T	(I-\ejth A^T)^{-1}	A^TPB \\
 & + B^T PA (\ejth I-A)^{-1} B.
 \end{align*}
 Now, using the anti-derivative of  $(\ejth I-  A)^{-1}$ and the fact that integration  (w.r.t.\ $\theta$) 
 of the complex conjugate is the conjugate of integration,  we have the result.
\end{proof}
The proof of the \thmref{dtn:thm:stable} is essentially similar 
to the proof of \cite[Theorem 3.8]{peter} without the need for inversion of the $A$ matrix.  
Using \cite[Theorem 1.31]{higham}, we have
\begin{align*}
\|\Gt\|_{[-\pi,\pi]}^2 = \tr B^T P B  = \|\Gt\|_{\Ltwosys}^2.
\end{align*}

\section{General case}\label{dtn:sec:gen}

In the previous section, the poles of a discrete-time system must be in the unit circle. We know that integration of a meromorphic function is possible 
as long as we are not integrating over a pole. Hence, systems with poles on the unit circle as well as inside and outside of the unit circle
(apart from poles within the limits of integration) would be a more general case.
This section is about the integration of a transfer function given in the descriptor form (see \eqref{dtn:des}) for the general case.
This will further help in obtaining an expression for the    frequency  truncated  discrete-time system norm in the general case.
It is assumed that eigenvalue of the matrix pencil $(A,E)$ can lie on the unit circle as well as inside and outside of the unit circle.  
The result needs logarithm of matrices as expected. However, there are few mathematical technicalities which we have to take care.

The first issue is  $\int (\ejth E-A)^{-1} \dth$ is a function of two matrices $E$ and $A$.
Hence, the definition  of a matrix function given in \cite{higham} does not help here as it is.
If $A$ or $E$ matrix is invertible 
then $\int (\ejth E-A)^{-1} \dth$ can be written as a function of $EA^{-1}$ or $E^{-1}A$ respectively.
However, the situation is a little complicated if both $E$ and $A$ are singular. Things can be simplified if we assume 
that the matrix pencil $(A,E)$ is regular.
To illustrate this further, assume that the $(A,E)$ is regular and $\alpha \ne 0$ then 
\begin{align*}
(zE-A)^{-1} &= \alpha (\alpha E + \beta A)^{-1} 
(z(I - \beta Q) - \alpha Q)^{-1} 
\end{align*}
where $Q:=A (\alpha E + \beta A)^{-1}$. 
Hence, we need $\int \alpha (z(I - \beta Q) - \alpha Q)^{-1}  \dth$ which is a function of just one matrix 
$Q$. If $\alpha=0$ then $A$ is invertible. This case has been  discussed already.

The second issue is related to the principal logarithm  of a matrix as it does not exist if eigenvalues of the matrix lie 
on the closed negative real axis $\cnR$. The critical task here is to choose 
the right anti-derivative of $\alpha (z(I - \beta Q) - \alpha Q)^{-1}$ such that the principal logarithm 
is defined. Furthermore, $Q$ can have eigenvalues on the unit circle as well as inside and outside of 
the unit circle. Hence, obtaining the right anti-derivative is quite challenging. 
For example, assume $\beta=0$, $\alpha=1$ and $A$ has at least one eigenvalue outside the unit circle.
In this case, $Q=A$ and $\log(I - \e^{-\jth} A)$ (which we obtained in the stable case) is not a valid anti-derivative of  $(z I - A)^{-1}$ as 
there exists a value of $\theta \in [-\pi,\pi]$ where eigenvalues of $I - \e^{-\jth} A$ lie on $\cnR$.
In this work, the right anti-derivative is obtained by the well-known tangent half-angle substitution i.e.\
\begin{align}\label{dtn:subg}
\ejth = \frac{1+\rj t}{1-\rj t} = -\frac{t-\rj }{t +\rj }
\end{align}
Here  $t=\tan(\frac{\theta}{2})$. 
Selection of the right anti-derivative is also an issue in \cite{hstn} which was solved by  taking $\rj$ out of the integrand whenever necessary. 
This technique is also used here along with the half-angle substitution.
\begin{theorem} \label{dtn:thm:main}
	Assume that a discrete-time system $K$ can be represented in the descriptor form as 
	$K(z)=(zE-A)^{-1}$ with real matrices $A$ and $E$.
	Also, assume that $\theta_1,\theta_2 \in (-\pi,\pi)$ and  the matrix pencil $(A,E)$ is regular i.e.\ $W:=\alpha E + \beta A$ is invertible for at-least one set of 
	 complex numbers
	  $\alpha$ 	and $\beta$.
	If $\e^{\rj \phi}$  not an  eigenvalue of $(A,E)$ for any $\phi\in\mb{R}$ and $\wrap(\phi) \in [\theta_1,\theta_2]$ then
	\begin{subequations}
		\label{dtn:main}
	\begin{align}
    &\int_{\theta_1}^{\theta_2} {K}(\ejth)\dth \notag\\
	&= \lim_{\eps\to 0}\frac{1}{ \rj}A_\eps ^{-1}\left(  -\eta I + 
	\log\left(\Gamma(\theta_2,\eps)\Gamma(\theta_1,\eps)^{-1}\right)  \right) \label{dtn:main1}\\
	&= \lim_{\eps\to 0}\frac{1}{ \rj}  \left(  -\eta I + 
	\log\left(\Gamma(\theta_1,\eps)^{-1}\Gamma(\theta_2,\eps)\right)  \right)A_\eps ^{-1} \label{dtn:main2}
	\end{align}
	\end{subequations}
	where 	$A_\eps :=A+ \alpha\eps I$, $E_\eps:=E-\beta\eps I$, $\eps\in\mb{C}$, $\eta:=\log\left(\frac{\tan(0.5\theta_2)-\rj}{\tan(0.5 \theta_1)-\rj}\right)$ and 
    $\Gamma(\theta,\eps) := (E_\eps +A_\eps )\tan(0.5\theta) -\rj (E_\eps- A_\eps ) $.  
\end{theorem}
\begin{proof} 
 Assume that $\alpha \ne 0$. Define
 \begin{align*}
 f_d(z,\theta) := \begin{cases}
 \frac{1}{\rj z}\left( -\eta + \log\frac{\Omega(z,t)}{\Omega(z,t_1)} \right), & z\ne 0 \\
 \frac{\alpha }{\rj}(\e^{-\jth_1} -\emjth) & z=0
 \end{cases}
 \end{align*}
 where $t:= \tan(\frac{\theta}{2})$ and $t_1:= \tan(\frac{\theta_1}{2})$ and $\Omega(z,t):=t (1-\beta z +\alpha z)- \rj ( 1-\beta z -\alpha z)$ for a real $t$.
 If $z$ is an eigenvalue of $Q:=A (\alpha E + \beta A)^{-1}$ then $\frac{z \alpha}{1- \beta z }$ is also an eigenvalue of $(A,E)$.
Hence, for all $\theta\in [\theta_1,\theta_2]$, $f_d(z,\theta)$ is defined on $\sigma(Q)$  as long 
 as $\e^{\rj \phi}$ not an eigenvalue of $(A,E)$ for any	$\wrap(\phi) \in [\theta_1,\theta_2]$ (see \thmref{dtn:thm:fd}).
	 	 Also,
	 \begin{align*}
	 \pd{f_d(z,\theta)}{\theta} &=  \alpha  \frac{1}{\ejth (1-\beta z) -\alpha z}. 
	 \end{align*}
	 	 Now, it follows from \thmref{dtn:thm:fi} and \cite[Theorem 6.2.27]{horn} that
	 $ \int_{\theta_1}^{\theta_2} {K}(\ejth)\dth = W^{-1} \int_{\theta_1}^{\theta_2} \alpha (z(I - \beta Q) - \alpha Q)^{-1}  \dth =  W^{-1} f_d(Q,\theta_2)$.
	 From \cite[Theorem 3.8]{higham} and \thmref{dtn:thm:fd}, the Fr\'{e}chet derivative $L_{f_d}(Q,X)$ of $f_d$ at $Q$ in the direction $X$ exists. Hence,
	 $f_d(Q+ \eps\alpha W^{-1}) = f_d(Q) + \eps L_{f_d}(Q,\alpha W^{-1}) +  o(|\eps\alpha|\|W^{-1}\|)$. 
	 Hence, we have \eqref{dtn:main1}. Equation \eqref{dtn:main2} follows from \cite[Theorem 1.13.c]{higham}. 
	 
	 On the other hand, if $\alpha =0$ then $A$ is invertible and $\beta \ne 0$ by the regularity of $(A,E)$.
	 Hence, $ \int_{\theta_1}^{\theta_2} {K}(\ejth)\dth = A^{-1} \int_{\theta_1}^{\theta_2} (\ejth EA^{-1}-I)^{-1}\dth$.
	 Define 
	 for a complex $z$
	 \begin{align*}
	 f_i(z,\theta) := 
	 \frac{1}{\rj}\left( -\eta + \log\frac{t(z+1) -\rj (z-1) }{t_1(z+1) -\rj (z-1)} \right)
	 \end{align*}
	 where $t:= \tan(\frac{\theta}{2})$ and $t_1:= \tan(\frac{\theta_1}{2})$.  
	  If $z$ is an eigenvalue of $EA^{-1}$  then $\frac{1}{z}$ is an eigenvalue of $(A,E)$.
	 Hence, for all $\theta\in [\theta_1,\theta_2]$, 
	 $f_i(z,\theta)$ is defined on all eigenvalues of $EA^{-1}$  as long 
	 as $\e^{\rj \phi}$ not an eigenvalue of $(A,E)$ for any	$\wrap(\phi) \in [\theta_1,\theta_2]$ (see \thmref{dtn:thm:fi}).
	 	  Now, it follows from \thmref{dtn:thm:fi} and \cite[Theorem 6.2.27]{horn} that
	  $ \int_{\theta_1}^{\theta_2} {K}(\ejth)\dth = A^{-1} \int_{\theta_1}^{\theta_2}  (\ejth EA^{-1}-I) \dth = A^{-1} f_i(EA^{-1},\theta_2)$. Equivalence to the limits can be proved in a manner similar the $\alpha \ne 0$ case.

	\end{proof}

 Equation \eqref{dtn:main} can be extended for $\theta=\pi$ as shown in the following 
result.
\begin{corollary}\label{dtn:coro:inf}
	Let $k(z)$ be as in \thmref{dtn:thm:main}. 
	Assume that $\theta_1 \in (-\pi,\pi]$ and  the matrix pencil $(A,E)$ is regular i.e.\ $W:=\alpha E + \beta A$ is invertible for at-least one set of 
	complex numbers
	 $\alpha$ 	and $\beta$.
	 If $\e^{\rj \phi}$  not an  eigenvalue of $(A,E)$ for any $\phi\in\mb{R}$ and $\wrap(\phi) \in [\theta_1,\pi]$ then
	 $\int_{\theta_1}^{\pi} {K}(\ejth)\dth$ equals
	 	 	 	\begin{align}
	 	\label{dtn:inf}
	 	 	&\lim_{\eps\to 0}\frac{1}{ \rj}A_\eps ^{-1}\left( \eta_f I - 
	 	\log\left(t_1 I -\rj \Phi_{-}(\eps)\Phi_{+}(\eps)^{-1}\right)  \right) \notag\\
	 	&= \lim_{\eps\to 0}\frac{1}{ \rj}  \left(  \eta_f I - 
	 	\log\left(t_1 I- \rj \Phi_{+}(\eps)^{-1}\Phi_{-}(\eps)\right)  \right)A_\eps ^{-1}
	 	\end{align}
	 	 where $t_1:= \tan(\frac{\theta_1}{2})$,	$A_\eps :=A+ \alpha\eps I$, $E_\eps:=E-\beta\eps I$, $\eps\in\mb{C}$, $\eta_f:=\log(t_1-\rj)$,
	 $\Phi_{+}(\eps) = E_\eps + A_\eps$ and $\Phi_{-}(\eps) = E_\eps - A_\eps$. 
\end{corollary}
\begin{proof}
Since $-1$ is not an eigenvalue of the matrix pencil $(A,E)$, $E+A$ is invertible. 
Define $t:= \tan(\frac{\theta}{2})$, $Q_\eps:=A_\eps W^{-1}$ and $\tilde{Q}:=(I-\beta Q_\eps   - \alpha Q_\eps  )(I-\beta Q_\eps   + \alpha Q_\eps  )^{-1}$.
Now, 
\begin{align*}
\log\left(t I -\rj (E_\eps - A_\eps )(E_\eps +A_\eps )^{-1}\right) = \log\left(t I -\rj \tilde{Q}\right)
\end{align*}
For sufficiently small $|\eps|$, the above logarithm exits if  $\e^{\rj \phi}$  not an  eigenvalue of $(A,E)$ for any $\phi\in\mb{R}$ and $\wrap(\phi) \in [\theta_1,\pi]$ (the proof is similar to the proof \thmref{dtn:thm:fd}).
Hence, \cite[Theorem 11.3,Theorem 11.2]{higham} implies that 
\begin{align*}
\log\left(\frac{t-\rj}{t_1-\rj}\right) &= \log(t-\rj) -\log(t_1-\rj)\\
\log( Q_t)&=\log(t I -\rj \tilde{Q}) -\log(t_1 I -\rj \tilde{Q}).
\end{align*}
where $Q_t:=(t I -\rj \tilde{Q}) (t_1 I -\rj \tilde{Q})^{-1}$.
Now, using
$\int_{\theta_1}^{\pi} {K}(\ejth)\dth = \lim_{\theta \to \pi}\int_{\theta_1}^{\theta_2} {K}(\ejth)\dth$ 
the result follows from \cite[Lemma 5(2)]{hstn}.
\end{proof}
Similar results can be obtained if integral limits are $[-\pi,\theta_1]$ or $[-\pi,\pi]$.
Note that if $A$ is invertible then \eqref{dtn:main} and \eqref{dtn:inf} can be simplified by taking 
$\eps=0$. Otherwise, \eqref{dtn:main} and \eqref{dtn:inf} needs a proper limit.
\secref{dtn:sec:comput} contains other forms of \eqref{dtn:main} and \eqref{dtn:inf} which are independent of any limit. However,
the current form is useful in obtaining a simple expression for the  frequency truncated discrete-time system norm
as explained in the following result.
\begin{theorem} 
 	Suppose a discrete-time system $\Gt$ can be represented in state-space as $\Gf(z)=C(zI-A)^{-1}B$ with real matrices A, B, and  C. 
 	Define
 	\begin{align*}
 		A_h &:= \begin{bmatrix}  A& 0 \\ C^T C  & I \end{bmatrix},  \quad
 		E_h:= \begin{bmatrix}  I & 0\\ 0  &A^T\end{bmatrix}, \quad\\
 		C_h&:=\begin{bmatrix} 0 & -B^T \end{bmatrix}, \quad \text{and}\quad
 		B_h:=\begin{bmatrix} B \\ 0 \end{bmatrix}.
 	\end{align*}
 
 	Assume that $\theta_1,\theta_2 \in (-\pi,\pi)$.
 	If $\e^{\rj \phi}$ not an eigenvalue of $A$ for any $\phi\in\mb{R}$ such that $wrap(\phi) \in [\theta_1,\theta_2]\cup[-\theta_1,-\theta_2]$ 
 	then
 	\begin{align*}
  		\int_{\theta_1}^{\theta_2}  \Gf^\sim(\ejth) \Gf(\ejth) \dth = 
 		\frac{1}{ \rj}C_h
 		\log\left(\Gamma_d(\theta_1)^{-1}\Gamma_d(\theta_2)\right) B_h
 	\end{align*}
 		where 
 		$\Gamma_d(\theta) := (E_h+ A_h)\tan(0.5\theta) -\rj (E_h-  A_h) $.
 \end{theorem} 
 \begin{proof}
 	The system 
 	$\Gf^\sim(z) \Gf(z)$ can be expressed as
 	 	$\Gf^\sim(z) \Gf(z)=  z C_h (zE_h - A_h)^{-1} B_h$.  	 	Clearly,
 	\begin{align*}
 	\int_{\theta_1}^{\theta_2} \Gf^\sim(\ejth) \Gf(\ejth) \dth
 	&= \int_{\theta_1}^{\theta_2} C_h \left(E_h - \emjth A_h \right)^{-1} B_h  \dth \\
 	&=  \int_{-\theta_1}^{-\theta_2} C_h \left( \ejth A_h - E_h\right)^{-1} B_h  \dth  	\end{align*}
  Assume that $\lambda_{max}$ and $\lambda_{min}$ represents the maximum and minimum absolute values of eigenvalues of 
 $A$. Now, $\det(E_h-\mu A_n)= \det(I- \mu A)\det (A^T -\mu I) \ne 0$
 if $\mu$ and $1/\mu$ is not an eigenvalue of $A$.
  Hence, the matrix pencil	$(E_h,A_h)$ is regular.
 It also shows that the eigenvalues of  $(E_h,A_h)$ are eigenvalues of $A$ and $A^{-1}$. This means,
  for any $\phi \in [\theta_1,\theta_2]$,  
 if $\e^{\rj \phi}$ is not an  eigenvalue of $(A_h,E_h)$  then   
 $\e^{-\rj \phi}$ not an  eigenvalue of $(A_h,E_h)$. This implies 
 $\e^{\rj \phi}$ not an eigenvalue of $A$ for any	$\phi \in [\theta_1,\theta_2]\cup[-\theta_1,-\theta_2]$.
Note that  $\lim_{\eps \to 0}(E_h + \eps \alpha I)^{-1}B_h = B_h$ and $C_hB_h=0$.
Finally, 
equivalence to the limit needed in \eqref{dtn:main2} can be proved in a manner given in the proof of 
\thmref{dtn:thm:main}.
\end{proof}
Note that the above result does not need any limits.

\section{Computation}
\label{dtn:sec:comput}

Equation \eqref{dtn:main} can be converted into another form (independent of $\eps$) which uses 
exponential of matrices  and 
\begin{align}\label{dtn:psi1}
\psi_1(A):= \sum_{j=0}^{\infty} \frac{1}{(j+1)!} A^j  \end{align}
The advantage is that both of these functions have numerically accurate and reliable implementation \cite[\S 10.5, \S 10.7.4]{higham}.

\begin{theorem} \label{dtn:thm:exp}
	Let $\psi_1$ be as in \eqref{dtn:psi1}.
	Using notations and conditions of \thmref{dtn:thm:main}, we have that
	$\int_{\theta_1}^{\theta_2} {K}(\ejth)\dth$ equals
		\begin{align}\label{dtn:exp}
		\frac{1}{\rj} W^{-1}\left(-\alpha L (\e^{-\rj \theta_2} I -\e^{-\rj \theta_1} \e^Y)  + \beta Y  \right)
		\end{align}
    where $Y_\eps :=   -\eta I + \log\left(\Gamma(\theta_1,\eps)^{-1}\Gamma(\theta_2,\eps)\right)$,
		$Y := \lim_{\eps \to 0} Y_\eps = -\eta I + \log\left(\Gamma(\theta_2,0)\Gamma(\theta_1,0)^{-1}\right)$ 		and $ L :=\lim_{\eps \to 0} (\e^{Y_\eps}-I)^{-1} Y_\eps = \psi_1(Y)^{-1}$.
\end{theorem}
\begin{proof} 
$Y =\lim_{\eps \to 0} Y_\eps$ due to \thmref{dtn:thm:fd}.

Note that
\begin{align}\label{dtn:tth}
(E_\eps + A_\eps  )t -\rj (E_\eps -A_\eps  ) = (t-\rj)(E_\eps - \emjth A_\eps  )
\end{align}
where $t:= \tan(\frac{\theta}{2})$. Now,
\begin{align*}
\exp\left( -\eta I + \log \left(\Gamma(\theta_2,\eps)\Gamma(\theta_1,\eps)^{-1})
\right) \right) = \e^{Y_\eps}.
\end{align*}
Using \cite[Theorem 10.2,Theorem 1.17]{higham}, we have
\begin{align*}
\frac{t_1-\rj}{t_2-\rj}\left( \Gamma(\theta_2,\eps)\Gamma(\theta_1,\eps)^{-1})
\right) =\e^{Y_\eps}
\end{align*}
where $t_i:= \tan(\frac{\theta_i}{2})$.
Assume $\alpha \ne 0$. 
Using \eqref{dtn:tth}, we have
\begin{align*}
\alpha (1-\e^{Y_\eps})E_\eps W^{-1} &=  \alpha(\e^{-\rj \theta_2} I - \e^{-\rj \theta_1} \e^{Y_\eps} )A_\eps W^{-1} \\
(1-\e^{Y_\eps})\alpha (E -\beta \eps I)W^{-1}  &= \alpha M  A_\eps  W^{-1} 
\end{align*}
where $W :=\alpha E + \beta A $ and $M_\eps :=\e^{-\rj \theta_2} I - \e^{-\rj \theta_1} \e^{Y_\eps} $.
Further, simplifying using $\alpha EW^{-1} = I-\beta A W^{-1}$, $Q_\eps := A_\eps W^{-1} =  AW^{-1} + \alpha \eps W^{-1}$
and $\alpha E_\eps W^{-1}:= I-\beta Q_\eps$, we have
$I-\e^{Y_\eps} =  (\alpha M_\eps  +  (I-\e^{Y_\eps})\beta ) Q_\eps$.
Using \cite[Theorem 1.13.a]{higham} and
\begin{align*}
\int_{\theta_1}^{\theta_2}K(\ejth) d\theta 
= \lim_{\eps \to 0}	\frac{1}{ \rj} W^{-1} Q_\eps^{-1} Y_\eps,
\end{align*}
we have the result.
Equivalence of $L$ and $\psi$ is standard \cite[\S 10.7.4]{higham}. 
Note that $L$ is invertible because it has no zero eigenvalues. 

If $\alpha =0$ then $A$ is invertible and $\beta \ne 0$ by the regularity of $(A,E)$. Hence,
$	\int_{\theta_1}^{\theta_2} {K}(\ejth)\dth
= \frac{1}{\rj} A^{-1} Y $. 
\end{proof}

Equation \eqref{dtn:inf} can be also modified in a manner similar to \thmref{dtn:thm:exp}.
\begin{corollary} \label{dtn:coro:exp}
	Let $\psi_1$ be as in \eqref{dtn:psi1}.
	Using notations and conditions of \cororef{dtn:coro:inf}, we have that
	\begin{align*}
	\int_{\theta_1}^{\pi} {K}(\ejth)\dth
	&= \frac{1}{\rj} W^{-1} L_f \left(\e^Y(\alpha  -\beta) + (\e^{-\rj \theta_1}\alpha + \beta)I \right) 
	\end{align*}
	where $Y_\eps :=  -\eta_f I + \log\left(t I -\rj\Phi_{-}(\eps)\Phi_{+}(\eps)^{-1}\right)$,
		$Y := \lim_{\eps \to 0} Y_\eps = -\eta_f I + \log\left(t I -\rj\Phi_{-}(0)\Phi_{+}(0)^{-1}\right)$ 		and $ L_f :=\lim_{\eps \to 0} (\e^{Y_\eps}-I)^{-1} Y_\eps = \psi_1(Y)$.
\end{corollary}

\section{A brief note on the continuous time system}\label{dtn:sec:cont}
Similar to \thmref{dtn:thm:main} and \thmref{dtn:thm:exp}, the results of \cite{hstn} can be extended to the continuous time descriptor systems.
These results are useful in the model reduction applications \cite{imran}.
The proof is similar to \thmref{dtn:thm:main} and \thmref{dtn:thm:exp}.
\begin{theorem} \label{dtn:thm:mainc}
	Let $\psi_1$ be as in \eqref{dtn:psi1}.
	Assume that a continuous time system $K$ can be represented in the descriptor form as 
	$K(s)=(sE-A)^{-1}$ with real matrices $A$ and $E$. Also, assume that $\omega_1,\omega_2 \in \mb{R}$ and
	the matrix pencil $(A,E)$ is regular i.e.\ $W:=\alpha E + \beta A$ is invertible for at-least one set of  complex numbers
	$\alpha$ and $\beta$.
	If the matrix pencil $(A,E)$ has no imaginary eigenvalue $\rj\lambda$
	 with $\lambda \in [\omega_1 , \omega_2 ]$  then
	\begin{align*}
		\int_{\omega_1}^{\omega_2} {K}(\jw)\dw 
		&= \lim_{\eps\to 0}\frac{1}{ \rj}E_\eps ^{-1}
		\log\left(\tilde{\Omega}(\omega_2,\eps)\tilde{\Omega}(\omega_1,\eps)^{-1}\right) \\
		&= \lim_{\eps\to 0}\frac{1}{ \rj}  
		\log\left(\tilde{\Omega}(\omega_1,\eps)^{-1}\tilde{\Omega}(\omega_2,\eps) \right) E_\eps ^{-1} \\
		&= - W^{-1}\left(\beta \tilde{L} (\omega_2 I - \e^{\tilde{Y}}\omega_1) + \rj \alpha \tilde{Y}\right)
		\end{align*}
	where 	$A_\eps :=A+ \alpha\eps I$, $E_\eps:=E-\beta\eps I$, $\eps\in\mb{C}$,
	$\tilde{\Omega}(\omega,\eps):=\omega E_\eps +\rj A_\eps$,
	$ \tilde{Y}_\eps:= \log\left(\tilde{\Omega}(\omega_2,\eps)\tilde{\Omega}(\omega_1,\eps)^{-1}\right)$
	$ \tilde{Y}:= \lim_{\eps\to 0} \tilde{Y}_\eps = \log\left(\tilde{\Omega}(\omega_2,0)\tilde{\Omega}(\omega_1,0)^{-1}\right)$ 
	and $\tilde{L}:= \lim_{\eps \to 0} (\e^{\tilde{Y}_\eps}-1)^{-1} \tilde{Y}_\eps= \psi_1(\tilde{Y})^{-1}$.
\end{theorem}
\section{Conclusions}
Computation of the  frequency truncated discrete-time system norm arises in 
different signal processing and model reduction applications.
This paper contains expressions for integral of the transfer function of a discrete-time system given in the descriptor form.
The result for the descriptor system is used in obtaining the  frequency truncated norm of a discrete-time system in the general case. Simplified results in case of stable systems are also given in the paper.
Similar results for the continuous time systems given in the descriptor form,  are also mentioned briefly. 
 
\section*{Acknowledgements}                              
	The author would like to thank Prof.\ Nicholas J.\ Higham (The University of Manchester, UK)
	and Prof. R.\ Alam (Indian Institute of Technology Guwahati, India) for many useful suggestions.

\appendix
\section{Appendix}

\label{appen}

The results in this 
section explain when the functions required in the proof of \thmref{dtn:thm:main}
are defined. Note that an analytic function in domain $D$ is always defined at
all $z\in D$. 

\begin{theorem} \label{dtn:thm:fd}
	Let $\alpha$ and $\beta$ be any two complex numbers and $\theta_1,\theta \in (-\pi,\pi)$. 
	Define a complex function 
	\begin{align*}
	f_d(z,\theta) := \begin{cases}
	\frac{1}{\rj z}\left( -\eta + \log\frac{\Omega(z,t)}{\Omega(z,t_1)} \right), & z\ne 0 \\
	\frac{\alpha }{\rj}(\e^{-\jth_1} -\emjth) & \text{elsewhere}
	\end{cases}
	\end{align*}
	 where $\eta:=\log\left(\frac{\tan(0.5\theta)-\rj}{\tan(0.5 \theta_1)-\rj}\right)$, $t:= \tan(\frac{\theta}{2})$ and $t_1:= \tan(\frac{\theta_1}{2})$ and $\Omega(z,t):=t (1-\beta z +\alpha z)- \rj ( 1-\beta z -\alpha z)$ for a real $t$.
	Assume that $\alpha\ne 0$. Then, $f_d(z,\theta)$ is analytic in 
	$\mb{C} \backslash \hat{D}$ where $\hat{D}:=\{ z \in \mb{C} |  \frac{\alpha z}{1 -\beta z}= \ejth, \wrap(\phi) \in [\theta_1,\theta] \} $. 
	Here, $\mb{C} \backslash \hat{D}$ is an open set.
\end{theorem}
\begin{proof}
	It is straightforward to verify that $\eta$ is well defined.

	It is now shown that $\mb{C} \backslash \hat{D}$ is an open set.
	Assume $\beta =0$, then the result is trivial.
	Assume $\beta\ne 0 $ then $g(z):=\frac{\alpha z}{1 -\beta z}$ is continuous apart from the point $z=\frac{1}{\beta}$.
	Since the set $Y:=\{\e^{\rj \phi} | \phi \in [\theta_1,\theta]\}$ is closed  and $g^{-1}(Y)$ does not contain $z=\frac{1}{\beta}$,
	continuity of $g$ in $\mb{C} \backslash \{0\}$ implies that $\mb{C} \backslash \hat{D}$ is an open set.

To check whether  $\log\frac{\Omega(z,t)}{\Omega(z,t_1)}$ is well defined or not,
first assume that $1-\beta z\ne 0$.
Then,  $\Omega(z,t_1) = (t_1+\rj) (-\e^{\jth_1}(1-\beta z) + \alpha z)$.
Hence,  $\Omega(z,t_1)$ is invertible as $z\notin \hat{D}$.
Assume that the principal log does not exist for a  $z\notin \hat{D}$. This means 
\begin{align*}
\log\frac{\Omega(z,t)}{\Omega(z,t_1)}
&=  \frac{t (1+a)- \rj ( 1-a)}{t_1 (1+a)- \rj ( 1-a)} =-\rho
\end{align*}
for all $\rho\ge 0$. Here, $a:=\frac{\alpha z}{1-\beta z}$.
The above implies that $a$ must be on unit circle i.e. $a=\e^{\rj \psi}$.
If $a=-1$ (i.e.\ $\psi=\pi$) then it is trivial to see that the principal log exist.
On the other hand if $a=\e^{\rj \psi}$ and $a\ne -1$ then 
\begin{align*}
\frac{(1+a)t-\rj (1-a)}{(1+a)t_1-\rj (1-a)}
&= \frac{t-\tan{\frac{\psi}{2}}}{t_1-\tan{\frac{\psi}{2}}} = -\rho
\end{align*}
iff
$\tan{\frac{\psi}{2}} = \frac{\rho t_1 +t}{\rho+1}$.
This means $\psi \in [\theta_1,\theta]$. Hence, $z\in\hat{D}$. Contradiction. 
Therefore, the principal log exists for all  $z\notin \hat{D}$.

Now, assume $1-\beta z = 0$.
Then, $\beta\ne 0$, $z\ne 0$ and $\Omega(z,t_1) = (t_1+\rj) \alpha z $.
Hence,  $\Omega(z,t_1)$ is invertible as $\alpha \ne 0$.
Also,  $\log\frac{\Omega(z,t)}{\Omega(z,t_1)} = \log(\frac{t+\rj}{t_1+\rj})$.
It is straightforward to see that this logarithm exists.

The above analysis implies that
$\frac{1}{\rj z}\left( -\eta + \log\frac{\Omega(z,t)}{\Omega(z,t_1)} \right)$ is analytic on an open set $\mb{C}\backslash \hat{D}$ 
apart from $z=0$ where it has a removable singularity. 
Clearly, $\lim_{z\to 0} f_d(z,\theta) = f_d(0,\theta)$. 
Hence, $f_d(z,\theta)$ is analytic in $\mb{C}\backslash \hat{D}$ (see e.g.\ \cite[\S 16.20]{apostol}). 
\end{proof}

The proof of the following result is similar to the proof of \thmref{dtn:thm:fd}.

\begin{theorem}\label{dtn:thm:fi}
	Let $\eta$ be as in \thmref{dtn:thm:fd} and $\theta_1,\theta \in (-\pi,\pi)$.
	Define a complex function 
	\begin{align*}
	f_i(z,\theta) := 
	\frac{1}{\rj}\left( -\eta + \log\frac{t(z+1) -\rj (z-1) }{t_1(z+1) -\rj (z-1)} \right)
	\end{align*}
	where $t:= \tan(\frac{\theta}{2})$ and $t_1:= \tan(\frac{\theta_1}{2})$. 
	Now, $f_i(z,\theta)$ is analytic in  $\mb{C} \backslash \tilde{D}$ where $\tilde{D}:=\{ \e^{-\rj \phi} | \wrap(\phi) \in [\theta_1,\theta] \} $. 
	Here, $\mb{C} \backslash \hat{D}$ is an open set.
\end{theorem}

\newcommand{\noopsort}[1]{}

\end{document}